\documentclass[conference,letterpaper]{IEEEtran}

\usepackage[utf8]{inputenc} 
\usepackage[T1]{fontenc}
\usepackage{url}
\usepackage{ifthen}
\usepackage{cite}
\usepackage[cmex10]{amsmath} 

\usepackage{amsthm}
\usepackage{amsmath,amssymb,amsfonts}
\usepackage{graphicx}
\usepackage{xcolor}
\newtheorem{thm}{Theorem}

\def\logisticweight{W}

\def\noise{Z}

\def\LLR{\text{LLR}}
\def\reliability{\Gamma}

\def\partitionnumber{\rho}

\def\codebook{\mathcal{C}}
\def\softoutput{S}
\def\listset{\mathcal{L}}
\def\outcome{O}
\def\nonoverlappingsets{\chi}

\usepackage{eso-pic}
\AddToShipoutPictureBG*{%
  \AtPageLowerLeft{%
    \hspace{0.5\paperwidth}\makebox[0pt][c]{%
      \raisebox{0.4in}{%
        \parbox{6.5in}{\footnotesize\centering
          \copyright~2026 IEEE. Personal use of this material is permitted. Permission from IEEE must be obtained for all other uses, in any current or future media, including reprinting/republishing this material for advertising or promotional purposes, creating new collective works, for resale or redistribution to servers or lists, or reuse of any copyrighted component of this work in other works.
        }%
      }%
    }%
  }%
}

\begin{document}
\title{Leveraging Code Structure to Improve Soft Output for GRAND, GCD, OSD, and SCL} 


\author{\IEEEauthorblockN{Jiewei Feng}
\IEEEauthorblockA{\textit{Dept. of Elec. \& Comput. Engineering} \\
\textit{Northeastern University}\\
Boston, USA \\
feng.ji@northeastern.edu}
\and
\IEEEauthorblockN{Ken R. Duffy}
\IEEEauthorblockA{\textit{Dept. of ECE \& Dept. of Mathematics} \\
\textit{Northeastern University}\\
Boston, USA \\
k.duffy@northeastern.edu}
\and
\IEEEauthorblockN{Muriel M{\'e}dard}
\IEEEauthorblockA{\textit{Research Laboratory for Electronics} \\
\textit{Massachusetts Institute of Technology}\\
Cambridge, USA \\
medard@mit.edu}
}

\maketitle

\begin{abstract}
In addition to a codeword, error correction decoders that provide blockwise soft output (SO) return an estimate of the likelihood that the decoding is correct. Following Forney, such estimates have traditionally only been possible for list decoders with the approximate soft output being the likelihood that a decoding is correct given it is assumed to be in the list. Recently, it has been established that Guessing Random Additive Noise Decoding can provide accurate soft output without the need for list decoding, whose idea can be applied to Guessing Codeword Decoding, Ordered Statistics Decoding, and Successive Cancellation decoding. Central to the improvement is the evaluation of an estimate of the likelihood that the correct decoding has not been found. We explore how binary linear codebook constraints can be leveraged to further enhance SO precision. We evaluate performance by adapting a forecasting statistic called the Brier Score. Results indicate that the SO generated by the approach is essentially as accurate as the maximum a posteriori estimate.
\end{abstract}


\section{Introduction}
For any soft input (SI) error correction decoder, it would be desirable if it returned blockwise SO in the form of an accurate a posteriori likelihood that the proposed decoding is correct in addition to a codeword. Blockwise SO could be leveraged to decide whether retransmission is necessary, or to calculate bitwise SO and so enable iterative decoding of long, powerful codes. The first blockwise SO formula was established by Forney \cite{Forney68}. It necessitates a list decoder as the SO is approximated by the conditional probability that each codeword in a list is correct given the transmitted codeword is assumed to be in the list. The same approximation underlies the celebrated Chase-Pyndiah iterative SISO decoder of turbo product codes \cite{pyndiah_1998} and subsequent alternatives \cite{bioglio2019construction,cocskun2024precoded,condo2022_iterative,galligan2023block}. 

Recently, it was established that SI GRAND algorithms \cite{solomon20, duffy2022_ordered, an2023soft, Duffy23ORBGRANDAI, Riaz24} can produce accurate blockwise SO, even for a single decoding \cite{galligan2023upgrade,yuan2025SOGRAND}. The key advance is the derivation of a dynamic estimate of the likelihood that the correct decoding has not been found. It can be readily evaluated during decoding and removes Forney's conditioning. That feature enables SO without list decoding while also improving SO for list decoding. Adapting the derivation of that estimate for GRAND, it has since been shown that accurate blockwise SO can be extracted from GCD \cite{ma2024guessing,zheng2024universal,2024ZhengLCGCD} or any version of OSD \cite{OSD95}, as well as SCL decoding \cite{tal2015list}, resulting in SOGRAND, SO-GCD, and SO-SCL, respectively, \cite{yuan2025SOGRAND,duffy2025SOGCD,yuan2025SOSCL}. Use of them as SISO decoders has shown, for example, that turbo product codes can outperform LDPC codes in standards while requiring fewer iterations to decode. More recently, these considerations have been extended to other decoders, e.g. \cite{janz2025soft}.

While a general formulation for the SO is provided in \cite{galligan2023upgrade,yuan2025SOGRAND}, the practical approximation used in \cite{yuan2025SOGRAND,duffy2025SOGCD,yuan2025SOSCL} assumes no specific codebook structure. Extending considerations, here we establish that it is possible to leverage binary linear codebook constraints to further enhance blockwise SO accuracy. Blockwise SO can be thought of as a forecast in the form of a predicted probability that a decoding is correct. To quantitatively assess the quality of SO, we use a variant of a metric from the forecasting literature called the Brier Score (BS) that enables empirical evaluation of the quality of probabilistic predictions from data \cite{Walter2006StrictProperRules}. 

The remainder of this paper is organized as follows. Section \ref{sec_preliminary} introduces decoding schemes and SO formulas. Section \ref{section_even_code} establishes the modification to SO that leverages binary linear codebook structure. Section \ref{sec:theory} provides theoretical analysis for the performance gain for ORBGRAND. Section \ref{section_bs} introduces Brier Score, its modification for SO decoders, and reports empirical results. Section \ref{sec_conclusion} provides final remarks.

\section{Soft output formulae}\label{sec_preliminary}

Consider a binary, non-necessarily linear, code, $\codebook\subset\{0,1\}^n$, consisting of $2^k$ binary strings of length $n$. Codewords, $X^n$, are assumed to be selected uniformly at random from the codebook. A modulated codeword $X^n$ is sent through a noisy channel and the detector receives the signal $R^n=(R_1,...,R_n)$. Given knowledge of the channel and assuming bits are impacted independently, the conditional pdf of $R$ given $X$, $f_{R|X}$, can be calculated. The log-likelihood ratio (LLR) of a received signal $r$ is then defined to be
\begin{align}
\LLR(r)=\log\left(f_{R|X}(r|1)/f_{R|X}(r|0)\right). \label{LLR_formula}
\end{align} 
From eq. \eqref{LLR_formula}, we set $\Gamma(R_i)=|\LLR(R_i)|$ to be the reliability of $R_i$. Defining $Y_i=(\text{sign}(\LLR (R_i))+1)/2$, $Y^n=(Y_1,Y_2,...,Y_n)$ denotes the hard demodulated bits extracted from $R^n$. With $\oplus$ representing addition in $\mathbb{F}_2$, defining $\noise^n=Y^n\oplus X^n$, $\noise^n$ represents the binary difference between the transmitted codeword and the hard decision bits. 

The likelihood that a hard decision bit $Y_i$ is in error can be identified as $B_i={e^{-\reliability(R_i)}}/(1+e^{-\reliability(R_i)})$, from which we can evaluate the posterior likelihood of a binary noise effect sequence, $z^n$, to be $p_{Z^n|R^n}(z^{n}|r^n) 
=\prod_{i=1}^n{(1-B_i)}\prod_{i:z_i=1}{\frac{B_i}{1-B_i}}$. Hence $p_{Z^n|R^n}(z^{n}|r^n)$ 
\begin{align}
\propto \prod_{i:z_i=1}{\frac{B_i}{1-B_i}}
=\exp\left({-\sum_{i=1}^n\reliability(r_i)z_i}\right).  \label{noiseprob}
\end{align}

Given $R^n$, decoders provide an estimate, $x^{n,*}\in\codebook$, of the transmitted $x^n$. 
A blockwise SO-decoder also provides an estimate of the probability that $x^{n,*}$ is the correct decoding, $\softoutput\in[0,1]$. The most accurate $S$ is the a posteriori likelihood that a given codeword $x^{n,*}\in\codebook$ is the transmitted one, which we dub the Maximum Likelihood (ML) SO, and is given by
\begin{align}
    p_{X^n|R^n}(x^{n,*}|r^n)=\frac{f_{R^n|X^n}(r^n|x^{n,*})}{\sum_{x^n\in\codebook}f_{R^n|X^n}(r^n|x^n)}\label{eq_trueSO},
\end{align}
where $f_{R^n|X^n}$ is the conditional pdf of the received signal $R^n$ conditioned on the transmitted codeword $X^n\in\codebook$. As the denominator has $2^k$ terms, however, it is not computationally viable to calculate for values of $k$ used in practice. When a decoder provides a list of $L>1$ codewords $\listset\subset\codebook$, Forney \cite{Forney68} introduced the following approximation 
\begin{align}
        p_{X^n|R^n}(x^{n,*}|r^n)\approx\frac{f_{R^n|X^n}(r^n|x^{n,*})}{\sum_{x^n\in\listset} f_{R^n|X^n}(r^n|x^{n})}\label{eq_ForneySO},
\end{align}
where the term $\sum_{x^n\in\codebook\setminus\listset}f_{R^n|X^n}(r^n|x^{n})$ is assumed to be zero. The approximation is the conditional probability that $x^{n,*}$ is correct assuming the transmitted codeword is the list. 

GRAND algorithms \cite{duffy19GRAND} operate by generating putative noise effects $z^{n,1},z^{n,2},...\in\{0,1\}^n$, inverting their effect from the demodulated signal $y^n$, and querying if what remains is in the codebook. The first instance where this is correct is a maximum likelihood decoding if queries are from most likely to least likely. By analyzing that process, an estimate of eq. \eqref{eq_trueSO} has been developed \cite{galligan2023upgrade,yuan2025SOGRAND} that includes an approximation to the omitted term in eq. \eqref{eq_ForneySO}. Crucially, the derivation of the SO formula that results does not require that noise effects are queried in any order, but instead applies to any query order. The principle behind that development has been adapted to SCL \cite{yuan2025SOSCL} and GCD \cite{duffy2025SOGCD}, which can be used with any Ordered Statistics Decoding (OSD) variant, but is, perhaps, most readily understood in the  context of its original derivation.

GRAND generates putative noise effects and outputs a decoded codeword whenever $\hat{x}^{n,i}=z^{n,i}\oplus y^n\in\codebook$. Suppose $L$ codewords are identified at query numbers $q_1,q_2,...,q_L$ and the decoder stops querying further noise effects at the $q_L$-th query. Defining
$\phi(j)=p_{Z^n|R^n}(z^{n,j}|r^n)$, in the work \cite{galligan2023upgrade,yuan2025SOGRAND}, eq. \eqref{eq_trueSO} is approximated by $p_{X^n|R^n}(\hat{x}^{n,q_i}|r^n)$
\begin{align}
     \approx{\phi(q_i)}\bigg/\left[{\sum\limits_{j=1}^L \phi(q_j)+\left(1-\sum\limits_{j=1}^{q_L} \phi(j)\right)\left(\frac{2^k-1}{2^n-1}\right)}\right]. \tag{SO-I}\label{eq_SOGRAND}
\end{align}
The second term in the denominator approximates the omitted term in eq. (\ref{eq_ForneySO}), essentially by assuming that the unvisited codewords in $\codebook\setminus\listset$ are uniformly distributed among the noise effects that have not yet been queried. Variants of that approximation have been used in both \cite{yuan2025SOSCL,duffy2025SOGCD}.

\section{Leveraging Codebook Structure for SO}\label{section_even_code}
Prior work has considered the use of binary linear codebook structure to inform query orders for GRAND \cite{Rowshan22, Rowshan23, rowshan2023segmented,2024WangPCGRAND,Lukas2025Tree,Lukas2025SYGRAND} and for GCD \cite{2024ZhengLCGCD,griffin2024using}. Here we consider its use to inform revisions to eq. \eqref{eq_SOGRAND}. Essentially, knowledge of binary linear code constraints impacts the evaluation of $\phi(j)=p_{Z^n|R^n}(z^{n,j}|r^n)$ by adding conditioning and changing the likelihood of remaining noise-effects that could result in a codeword.

The principle is most readily introduced with an example. For any binary sequence $z^n$, define 
$\Phi(z^n) = \bigoplus_{i=1}^n z_i$ to be its parity. A code is called {\it even} if all codewords have even parity, i.e. $\Phi(X^n)=0$ for all $X^n\in\codebook$. Many well-known codes, including extended BCH (eBCH) codes, Polar codes, polarization-adjusted convolutional codes, as well as Koopman’s CRC codes \cite{koopman2004cyclic} with certain dimensions, are even.

For any $x^n$ chosen from an even code, $\Phi(Y_i)  = \Phi(x_i\oplus Z_i) 
    = \Phi(x_i) \oplus\, \Phi(Z_i)
    = \Phi(Z_i),$
as $\Phi(x_i)=0$. Thus the parity of the demodulated bits at the receiver, $Y^n$, must be the same as the parity of the noise-effect $Z^n$ that has impacted the codeword $x^n$. As a result, for an even code, only noise effects that have the same parity as $Y^n$ can result in the identification of a codeword. As a complexity reduction technique for GRAND, that results in the number of queries required to identify a codeword being reduced by a factor of up to $1/2$ without any sacrifice in precision. By leveraging more of the linear codebook structure, query numbers can be further reduced, e.g. \cite{Rowshan22, Rowshan23, rowshan2023segmented,2024WangPCGRAND}.

Whether one uses the even code property to inform query order or not, it can be exploited to enhance SO. For an even code one has that $p_{Z^n|R^n}(z^n|r^n)$
\begin{align*}
=
    \begin{cases}
    0 & \text{ if } \Phi(z^n)\neq  \Phi(y^n), \\
    \displaystyle 
    \frac{p_{Z^n|R^n}(z^n|r^n)}{P(\Phi(Z^n) = \Phi(y^n)|R^n=r^n)}
      & \text{ if } \Phi(z^n)= \Phi(y^n).
    \end{cases}
\end{align*}
That is, knowing that the code is even, noise effects that cannot lead to a codeword have no probability, while those that can lead to a codeword are inflated by the likelihood that the parity of the noise effect is consistent with the demodulated sequence. The latter can be calculated via Lemma 1 of \cite{Gallager62}, which is well-known from its use in the decoding of Low Density Parity Check Codes: $P(\Phi(Z^n) = \Phi(y^n)|R^n=r^n)$
\begin{align*}
    = 
    \begin{cases} 
      \frac{1}{2}\left(1+\prod_{i=1}^n(1-2B_i)\right) & \text{if } \Phi(y^n) = 0, \\
      \frac{1}{2}\left(1-\prod_{i=1}^n(1-2B_i)\right) & \text{if } \Phi(y^n) = 1.
   \end{cases}
\end{align*}
One further modification to eq. \eqref{eq_SOGRAND} is required. Eq. \eqref{eq_SOGRAND} effectively assumes that the $2^k-1$ non-transmitted codewords are uniformly distributed in the $2^n-1$ possible noise effect sequences. For an even code, the latter is changed to $2^{n-1}-1$ as half of all strings have the wrong parity. This results in an approximation to eq. (\ref{eq_trueSO}) for even codes using GRAND as expressed in eq. (\ref{eq_SO_even}). Define $I_{i}=1$ if the $i$-th queried noise effect has the same parity as $y^n$ and $I_{i}=0$ if otherwise, where the queries with $I_i=0$ are the ones that will not lead to codewords and hence can be skipped based on the even code property. Then, $p_{X^n|R^n}(\hat{x}^{n,q_i}|r^n)$
\begin{align}
        \approx\frac{\phi(q_i)}{\sum\limits_{j=1}^L \phi(q_j)+\left(\psi-\sum\limits_{j=1}^{q_L}\phi(j)I_{j}\right)\displaystyle\frac{2^k-1}{2^{n-1}-1}},\label{eq_SO_even}\tag{SO-II}
\end{align}
where $\phi(j)=p_{Z^n|R^n}(z^{n,j}|r^n)$ and 
\begin{align*}
2\psi=\begin{cases} 
      \left(1+\prod_{i=1}^n(1-2B_i)\right) & \text{if } \Phi(y^n)=0,  \\
      \left(1-\prod_{i=1}^n(1-2B_i)\right) & \text{if } \Phi(y^n)=1.
   \end{cases} 
\end{align*}

These ideas can also be applied to SO-GCD, SO-OSD, and SO-SCL. On the other hand, additional binary linear constraints of non-overlapping bits can be used in a similar fashion. For illustration, we present a generalization of eq. (\ref{eq_SO_even}). 
Let $z^{n,q_j}_b$ denote the $b$-th entry in $z^{n,q_j}$. Suppose there exist $\delta$ non-overlapping subsets of $\{1,\cdots,n\}$, denoted as $\nonoverlappingsets_1,\cdots,\nonoverlappingsets_\delta$, such that $\bigoplus_{b\in\nonoverlappingsets_a}x_b=0$ for any codeword $x^n\in\codebook$ and any $a\in\{1,\cdots,\delta\}$. This means that each subset $\nonoverlappingsets_a$ of the code forms an even code, implying that any noise effect $z^n$ that can lead to a codeword must satisfy $\bigoplus_{b\in\nonoverlappingsets_a} z_b=\bigoplus_{b\in\nonoverlappingsets_a} y_b$ for all $a$. These constraints have been leveraged in \cite{rowshan2023segmented,Rowshan22,Rowshan23} for guesswork reduction by only generating noise effects that satisfy these extracted constraints. Define $I'_j=1$ if the $j$-th queried noise effect $z^{n,j}$ satisfies $\bigoplus_{b\in\nonoverlappingsets_a}z^{n,j}_{b}=\bigoplus_{b\in\nonoverlappingsets_a}y_{b}$ for all $a\in\{1,\cdots,\delta\}$ and $I'_j=0$ otherwise, which indicates whether the noise effect satisfies the constraints on each of the subsets $\nonoverlappingsets_1,\cdots,\nonoverlappingsets_\delta$, and define
\begin{align*}
    2\psi_a=\begin{cases} 
      \left(1+\prod_{b\in\nonoverlappingsets_a}(1-2B_b)\right) & \text{if } \bigoplus_{b\in\nonoverlappingsets_a}y_b=0, \\
      \left(1-\prod_{b\in\nonoverlappingsets_a}(1-2B_b)\right) & \text{if } \bigoplus_{b\in\nonoverlappingsets_a}y_b=1.
   \end{cases}
\end{align*}
Given $\nonoverlappingsets_1,\cdots,\nonoverlappingsets_\delta$, (\ref{eq_SO_even}) is generalized to $p_{X^n|R^n}(\hat{x}^{n,q_i}|r^n)$
\begin{align}
    \approx\frac{\phi(q_i)}{\sum\limits_{j=1}^L \phi(q_j)+\left(\prod_{a=1}^\delta\psi_a-\sum_{j=1}^{q_L}\phi(j)I'_{j}\right)\displaystyle\frac{2^k-1}{2^{n-\delta}-1}}.\label{eq_SO_disjoint_Sets}
\end{align}

\section{SO Quality - Theory}
\label{sec:theory}
To theoretically assess the difference between the SO formula with and without knowledge of the code being even, eq. \eqref{eq_SOGRAND} and eq. \eqref{eq_SO_even}, we do so in the context of Ordered Reliability Bit Guessing Random Additive Noise Decoding (ORBGRAND) \cite{Duffy21_ordered,duffy2022_ordered} a near maximum likelihood decoder that has been shown to be capacity achieving \cite{Liuetal23}.
Note that eq. \eqref{eq_SOGRAND} and \eqref{eq_SO_even} only differ in the second term of denominator, which approximates the probability of unvisited codewords. Therefore, we only need to study the difference 
\begin{align*}
    &\left(\psi-\sum_{j=1}^{q_L}\phi(j)I_{j}\right)\frac{2^k-1}{2^{n-1}-1}-\left(1-\sum_{j=1}^{q_L}\phi(j)\right)\frac{2^k-1}{2^n-1}\\
    &\approx \left(\psi-\sum_{j=1}^{q_L}\phi(j)I_{j}\right)\frac{2^k}{2^{n-1}}-\left(1-\sum_{j=1}^{q_L}\phi(j)\right)\frac{2^k}{2^n}\\   
    &=\frac{2^k}{2^n}\left(2\psi-2\sum_{j=1}^{q_L}\phi(j)I_{j}-1+\sum_{j=1}^{q_L}\phi(j)\right)=:\frac{2^k}{2^n}\Delta,
\end{align*}
where the approximation in the second line assumes that $2^{n}>2^{n-1}>2^k\gg 1$. Thus it suffices to analyze the quantity $\Delta$ to investigate the difference between the two SO formulae.

For a summary of basic ORBGRAND,  let $\pi_i\in\{1,\ldots,n\}$ be the rank of the reliability, in increasing order, of the $i$-th bit within the $n$ received bits and let $\logisticweight(z^n)= \sum_{i=1}^n\pi_iz_i$ denote the Logistic Weight (LW) of the noise effect $z^n$ \cite{Duffy21_ordered,duffy2022_ordered}. Basic ORBGRAND assumes that the rank ordered reliabilities are approximately linear; namely, that there exists $\beta>0$ such that $\reliability(r_i)\approx\beta\pi_i$ for all $i=1,\cdots,n$. If rank ordered reliability is linear, then the right hand side of eq. \eqref{noiseprob} equals
\begin{align}
    \exp\left(-\sum_{i=1}^n \beta\pi_iz_i\right).\label{eq_temp1}
\end{align}
It has been shown by \cite{Alex2026LRC} that this linear reliability channel (LRC) model closely approximates a wide range of channel models. Hence, to create noise effect sequences in order of decreasing likelihood it suffices to create them in increasing order of their LW, which is what basic ORBGRAND does. The noise effects with LW $w$ are generated by finding all distinct partitions of the integer $w$ into components that are at most $n$, where each component of a partition indicates which bit should be flipped to generate a noise effect. Note that the correspondence between noise effects and distinct integer partitions with component at most $n$ is a bijection. 

The Pentagonal Number Theorem \cite{gupta1970partitions} states that
\begin{align}
    \prod_{i=1}^\infty(1-x^i)=1+\sum_{k=1}^\infty(-1)^{k}\left(x^{\frac{k(3k+1)}{2}}+x^{\frac{k(3k-1)}{2}}\right),\label{eq_pentagonal_thm}
\end{align}
which implies the number of distinct partitions of an integer $w$ into even parts equals the number of distinct partitions of $w$ into odd parts if $w$ is not of the form $k(3k\pm 1)/2$. The two quantities differ by exactly one if $w$ is of the form $k(3k\pm 1)/2$, which is called a generalized pentagonal number. Interpreted for LWs $w\leq n$, Eq. \eqref{eq_pentagonal_thm} says that the number of noise effects with even parity is the same as the noise effects with odd parity if $w$ is not a generalized pentagonal number and differs by 1 if $w$ is a generalized pentagonal number.

Let $\partitionnumber_0(w,n)$ denote the number of even distinct partitions where no component can exceed $n$, then $\partitionnumber_0(w,n)$ equals the number of queries with LW $w$ with even Hamming Weight. Similarly, $\partitionnumber_1(w,n)$ corresponds to the number of odd distinct partitions where no terms can exceed $n$, which further represents the number of queries with LW $w$ with odd Hamming Weight given the code length $n$. Eq. (\ref{eq_pentagonal_thm}) then says, for $w\leq n$,
\begin{align}
    \partitionnumber_0(w,n)-\partitionnumber_1(w,n)=\begin{cases}
        (-1)^k &\text{if } w=k(3k\pm 1)/2\\
        0 & \text{else}
    \end{cases}\label{eq_query_number}
\end{align}

For ease of presentation, we assume that if ORBGRAND finds the $L$-th most likely codeword at LW $w^*$, instead of stopping immediately, it stops after generating all noise effects with LW $w^*$. This assumption does not affect the conclusion solely simplifies the proof. In addition, we assume that $L$ codewords will be found within the noise effects with LW not exceeding $n$. 

\begin{thm}\label{thm_diff} 
    Assume the LRC with eq. \eqref{eq_temp1}.
    Suppose that ORBGRAND finds the $L$-th codeword at LW $w^*\leq n$. Define
    \begin{align}
    \Theta(w)=\exp{(-\beta w)}\prod_{i=1}^n(1-B_i).\label{eq_pl}
    \end{align}
    When $\beta\rightarrow 0$, which corresponds to the limit of low SNR, we have $\lim_{\beta\rightarrow 0} |\Delta|= 2^{-n}$. For a general value of $\beta$,  
    \begin{align*}
    |\Delta|\leq 2\Theta(w^*) +\left|\sum_{w=n+1}^{n(n+1)/2} \Theta(w)(\partitionnumber_0(w,n)-\partitionnumber_1(w,n))\right|
\end{align*}
\end{thm}
\begin{proof}
By eq. (\ref{eq_temp1}), $\phi(j)=p_{Z^n|R^n}(z^{n,j}|r^n)$ can be calculated by $\Theta(w)$ if the $j$-th generated noise effect has LW $w$. Then
\begin{align}
    \left|\Delta\right|=\left|
    1-2\psi+\sum_{j=1}^{q_L}\phi(j)I_j-\sum_{j=1}^{q_L}\phi(j)(1-I_j)
    \right|\nonumber\\
    \overset{(\ref{eq_pl})}{=}\left|1-2\psi\pm\sum_{w=0}^{w^*}\Theta(w)(\partitionnumber_0(w,n)-\partitionnumber_1(w,n))\right|\label{eq_lefterror1}.
\end{align}

As $\beta\rightarrow 0$, $B_i\rightarrow 1/2$, and hence $\Theta(w)\rightarrow2^{-n}$ and $|\Delta|\rightarrow 2^{-n}$. Letting $g(w,n)=\Theta(w)(\partitionnumber_0(w,n)-\partitionnumber_1(w,n))$, we get
\begin{align}
    &|\Delta|=\left|\left(1-\sum_{j=1}^{q_L}\phi(j)\right)-2\left(\psi-\sum_{j=1}^{q_L}\phi(j)I_j\right)\right|\nonumber\\
    &=\left|\left(1-\psi-\sum_{j=1}^{q_L}\phi(j)(1-I_j) \right)-\left(\psi- \sum_{j=1}^{q_L}\phi(j)I_j\right)\right|\nonumber\\
    &=\left|\sum_{w=w^*+1}^{n} g(w,n)+\sum_{w=n+1}^{n(n+1)/2} g(w,n)\right|\label{eq_lefterror2}.
\end{align}
Note that the noise effects are generated such that $\Theta(i)\leq\Theta(j)$ for $j\geq i$. Together with eq. \eqref{eq_query_number}, the first summation in eq. \eqref{eq_lefterror2} can be bounded by the first two terms with $w$ a generalized pentagonal number:
\begin{align}
    (\ref{eq_lefterror2})\leq |\Theta(k_1)\pm \Theta(k_2)|+\left|\sum_{w=n+1}^{n(n+1)/2} g(w,n)\right|\nonumber\\
    \leq 2\Theta(w^*) +\left|\sum_{w=n+1}^{n(n+1)/2} \Theta(w)(\partitionnumber_0(w,n)-\partitionnumber_1(w,n))\right|,\label{eq_decomposition}
\end{align}
where $k_1<k_2$ are the first two pentagonal numbers after $w^*$.
\end{proof}

Theorem \ref{thm_diff} establishes that, when $\beta\rightarrow 0$, the difference $2^{k-n}\Delta$ between the approximations of likelihood of unvisited codewords in eq. \eqref{eq_SOGRAND} and \eqref{eq_SO_even} converges to $2^{k-2n}$ which is negligible for practical $k$ and $n$. For a general value of $\beta$, the upper bound of $|\Delta|$ can be decomposed into two parts in eq. \eqref{eq_decomposition}. The term $2\Theta(w^*)$ vanishes as $w^*$ increases, which can be achieved by increasing list size. The other term is a constant independent of the stopping LW $w^*\leq n$, which is rare for moderate redundancy codes. This result suggests that while an improvement in accuracy with additional constraints is to be expected, it is likely to be modest. 

\section{SO Quality Quantification}\label{section_bs}
Blockwise SO can be considered as a forecast for the correctness of the decoding. With $\softoutput\in[0,1]$ denoting the blockwise SO, we define $\outcome=1$ if the decoded codeword is correct and $\outcome=0$ if not. Two features are desirable for a forecaster, which are called calibration and refinement \cite{Walter2006StrictProperRules}. A forecaster is well-calibrated if the $P(\outcome|\softoutput=s)=s$ holds for all $s\in[0,1]$. That is, the conditional likelihood of correctness given the forecast corresponds to the forecast. It has been demonstrated empirically that SOGRAND \cite{galligan2023upgrade,yuan2025SOGRAND}, SO-GCD \cite{duffy2025SOGCD} and SO-SCL \cite{yuan2025SOSCL} are well calibrated forecasters but Forney's approximation is not. A forecaster is called least-refined if no other well-calibrated forecaster can provide more confident predictions.  

Consider $N$ experiments whose forecasts and outcomes are $\{(s_t,o_t):t=1,\ldots,N\}$. The Brier Score \cite{Walter2006StrictProperRules} is given by
\begin{align}
    \text{BS}&=\frac{1}{N}\sum_{t=1}^N(s_t-o_t)^2.\label{BS1}
\end{align}
If $s$ only takes values in a finite set $\mathcal{S}\subset[0,1]$, then eq. \eqref{BS1} can be written in a form that separates terms for calibration and refinement. Let $v(s)=N^{-1}\sum_{t=1}^N I_{\{s_t=s\}}$ denote the empirical frequency that the forecaster gives the prediction $s$ and $\rho(s)=\sum_{t=1}^N I_{\{s_t=s\}}o_t/\sum_{t=1}^N I_{\{s_t=s\}}$ be the empirical probability of correct prediction conditioned on the prediction value being $s$, then
\begin{align}
    \text{BS}=\sum_{s\in\mathcal{S}} v(s)[s-\rho(s)]^2+\sum_{s\in\mathcal{S}}v(s)\rho(s)[1-\rho(s)].\label{BS3}
\end{align}
When a forecaster is well-calibrated, the first term is zero. When a forecaster is the least-refined, the second term is the minimum among all other well-calibrated forecasters. In general, BS allows the evaluation of forecaster performance. 

In most settings where BS values are used, the minimum of the likelihood that the outcome, $O$, is 0 or 1 is bounded away from 0. In error correction systems, that is not the case as most practical decoders return correct decodings most of the time. Considering a na\"ive estimator that is always confident that the decoding is correct, i.e. it always provides $S=1$. That forecast results in eq. \eqref{BS1} giving that BS = BLER. Consequently, we can expect that any decoder producing accurate SO will have a BS value at most the same as BLER, which is typically close to zero. Therefore, when comparing SO formulas, we consider them on a $\log$ scale. Moreover, in some cases, Brier Score from different groups can be almost identical due to these extremely rare cases of erroneous decoding. Inspired by the Brier Skill Score \cite{Mur73}, we propose a Brier Score Ratio (BSR) over na\"ive predictors as an alternative quantification. Before comparing the Brier Score of different decoders, each of them can be divided by the lowest BLER at that $E_b/N_0$. 

To assess the SO formulae, we consider binary modulation experiencing Additive White Gaussian Noise. We consider an eBCH code, which is even, and present results for dimensions $(16,11)$ as with $2^{11}=1024$ codewords direct evaluation of the ML SO in eq. \eqref{eq_trueSO} is possible. Similar results, which are not reported due to space constraints, were observed for longer codes and codes of distinct structure. Here, a max decoding with $L\geq 2$ means considering the element in the decoded list with highest SO. Through the lens of Brier Score, we first compare the recently developed SO with Forney's approach and the ML SO. We then consider the use of one codebook constraint, followed by the use of two.

Fig. \ref{fig_overall_eBCH_16_11} presents a comparison of different SO calculations when $L=2$. The dashed line indicates the na\"ive SO of ORBGRAND, its BLER. The line with diamonds represents the ML SO, eq. (\ref{eq_trueSO}). The line with squares represents Forney's approximation, eq. \eqref{eq_ForneySO}, where different decoders with the same decoded lists provide the same value. The line with crosses represents SO-GCD, while the line with circles represents ORBGRAND with eq. \eqref{eq_SOGRAND}. As can be seen, Forney's approximation has higher BS at lower SNR compared to other SO while SO-GCD and eq. \eqref{eq_SOGRAND} have almost identical BS to the ML SO, which is the most accurate estimate possible without knowing what codeword was transmitted.

\begin{figure}

    \centering
    \includegraphics[width=0.95\linewidth]{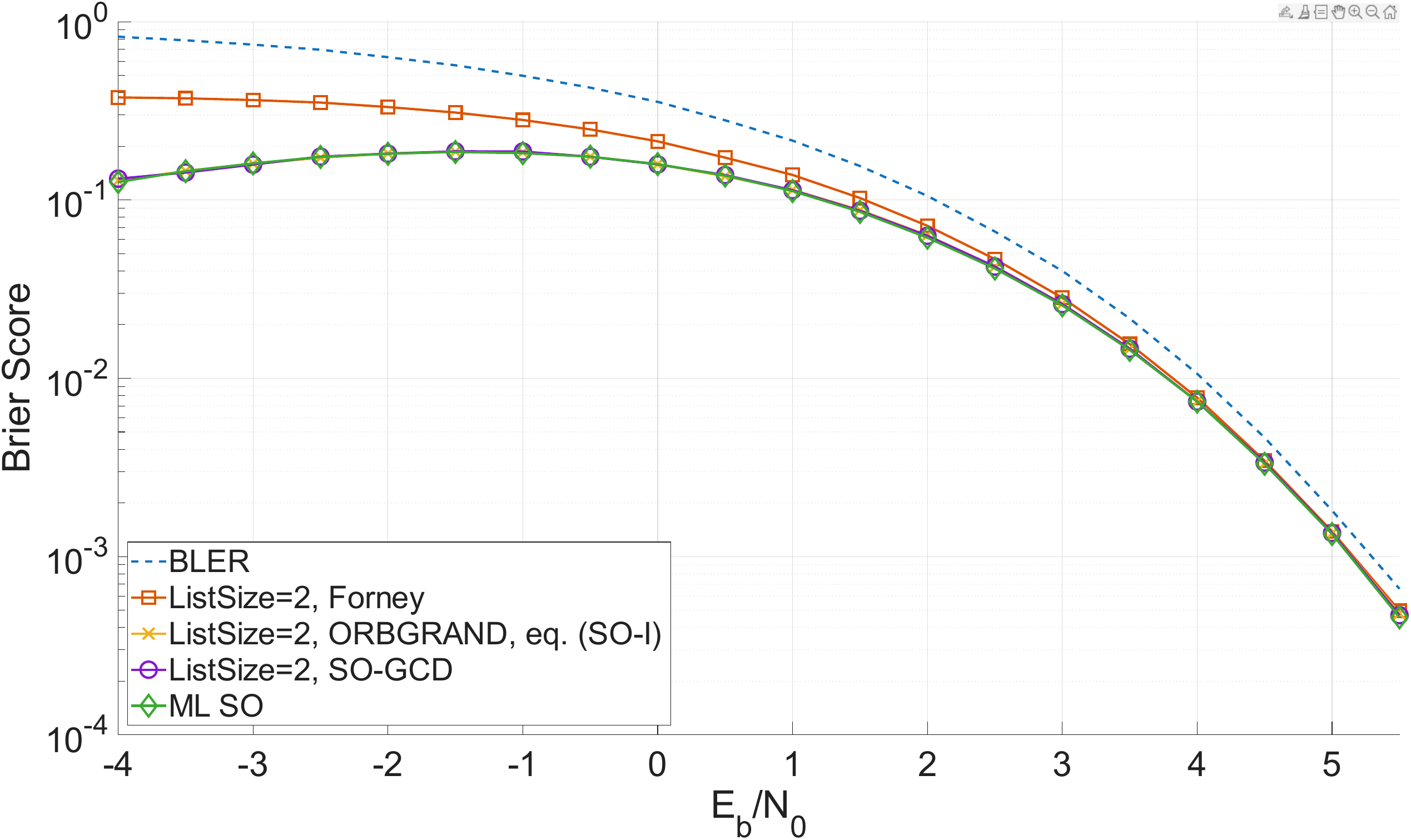}
    \caption{Example of Brier Score vs $E_b/N_0$ using eBCH(16, 11) code. } 
    \label{fig_overall_eBCH_16_11}
\end{figure}

\begin{figure}

    \centering
    \includegraphics[width=0.95\linewidth]{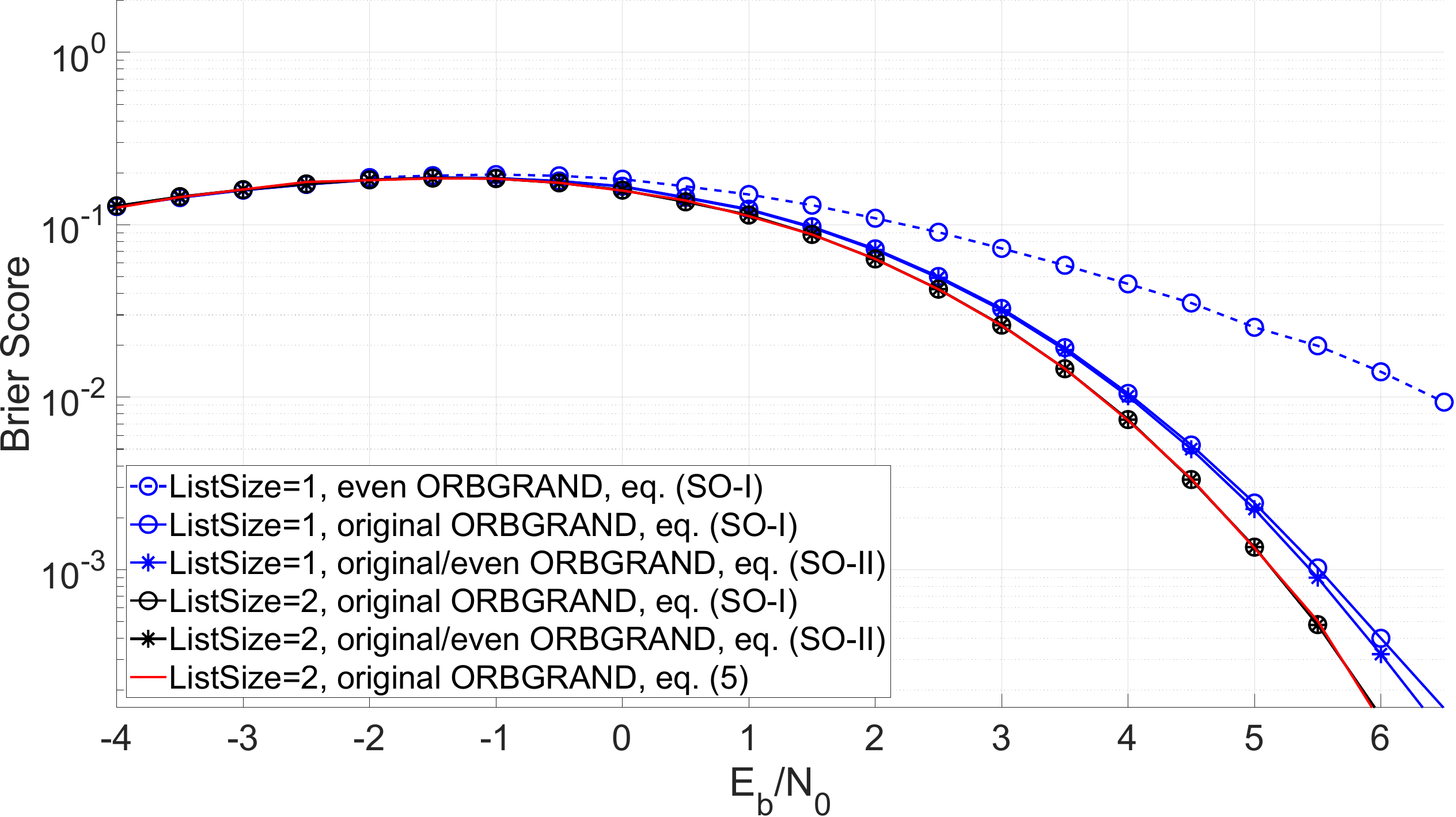}
    \caption{Examples of Brier Score using eBCH(16, 11) code. } 
    \label{eBCH_16_11_even_compare}
\end{figure}

Fig. \ref{eBCH_16_11_even_compare} primarily investigates the results in Section \ref{section_even_code} and \ref{sec:theory} that consider the even code constraint. Solid lines represent the standard ORBGRAND that does not skip queries. Blue lines represent $L=1$ while black lines represent $L=2$. Lines with circles represent the original calculation of SO, eq. \eqref{eq_SOGRAND}. Lines with stars represent eq. \eqref{eq_SO_even} that leverages the even code property. The same even code property can be used to inform query order as mentioned in Section \ref{section_even_code}, resulting in even ORBGRAND that has the same BLER with up to $50\%$ reduction in the number of queries. The queries with $I_i=0$ in eq. \eqref{eq_SO_even} are the queries skipped in even ORBGRAND and do not affect the calculation of eq. \eqref{eq_SO_even}. Therefore, even ORBGRAND with eq. \eqref{eq_SO_even} has the same Brier Score as original ORBGRAND with eq. \eqref{eq_SO_even}, where both results can be represented by the solid lines with stars. The dashed blue line represents even ORBGRAND with eq. \eqref{eq_SOGRAND}, which demonstrates the degradation of SO when the codebook structure leveraged in query order is not taken into account for SO calculation. In contrast, eq. \eqref{eq_SO_even} that leverages the even code property is compatible with even ORBGRAND and provides SO values without degradation. Moreover, eq. \eqref{eq_SO_even} provides better SO at $L=1$, but the improvement vanishes at $L=2$, which is consistent with the result in Section \ref{sec:theory}. 

In addition, Fig. \ref{eBCH_16_11_even_compare} includes a verification of calculating SO using multiple constraints via eq. \eqref{eq_SO_disjoint_Sets}. The red line represents original ORBGRAND with a list size of two and calculates SO via eq. (\ref{eq_SO_disjoint_Sets}) with $\delta=2$ constraints, where $\nonoverlappingsets_1=\{1,2,3,4,6,8,9,12\}$ and $\nonoverlappingsets_2=\{5,7,10,11,13,14,15,16\}$ are extracted using the first row of $H$ for the even eBCH$(16,11)$ code. As depicted by the figure, eq. (\ref{eq_SO_disjoint_Sets}) using two constraints has a similar Brier Score to eq. (\ref{eq_SOGRAND}), implying that eq. \eqref{eq_SO_disjoint_Sets} is compatible with decoders that use the same codebook structure to inform query order.

\section{Conclusions}\label{sec_conclusion}
We demonstrated that binary linear codebook constraints can be leveraged to improve the recently developed blockwise SO. Simulation results using Brier Score suggest that both SOGRAND and SO-GCD with list size two have almost identical accuracy to a ML SO.

\bibliographystyle{IEEEtran}
\bibliography{b}

@article{duffy2022_ordered,
  title={Ordered reliability bits guessing random additive noise decoding},
  author={Duffy, Ken R. and An, Wei and Medard, Muriel},
  journal={IEEE Trans. Signal Proc.},
  pages={4528-4542},
  volume={70},
  year={2022},
}

@ARTICLE{duffy19GRAND, 
author={K. R. {Duffy} and J. {Li} and M. {M\'edard}}, 
fjournal={IEEE Transactions on Information Theory}, 
Journal = {IEEE Trans. Inf. Theory},
title={Capacity-Achieving Guessing Random Additive Noise Decoding}, 
year={2019}, 
volume={65}, 
number={7}, 
pages={4023-4040}, 
}

@article{bioglio2019construction,
  title={Construction and decoding of product codes with non-systematic polar codes},
  author={Bioglio, Valerio and Condo, Carlo and Land, Ingmar},
  journal={IEEE Wireless Commun. Netw. Conf. (WCNC)},
  pages={1--6},
  year={2019}
}

@inproceedings{solomon20,
  title={Soft Maximum Likelihood Decoding using {GRAND}},
  author={A. Solomon and K. R. Duffy and M. M\'edard},
  booktitle={IEEE ICC},
  year={2020},
}

@article{yuan2025SOGRAND,
  title={{Soft-output (SO) GRAND and iterative decoding to outperform LDPC codes}},
  author={Yuan, Peihong and Medard, Muriel and Galligan, Kevin and Duffy, Ken R},
  journal={IEEE Trans. Wireless Commun.},
   year={2025},
  volume={24},
  number={4},
  pages={3386-3399},
}

@ARTICLE{rowshan2023segmented,
  author={Rowshan, Mohammad and Yuan, Jinhong},
  journal={IEEE Trans. Commun.}, 
  title={Segmented {GRAND}: Complexity Reduction through Sub-Pattern Combination}, 
  year={2025},
  volume={},
  number={},
  pages={1-1},
  doi={10.1109/TCOMM.2025.3541094}}

@ARTICLE{Forney68,
  author={Forney, G.},
  journal={IEEE Trans. Inf. Theory}, 
  title={Exponential error bounds for erasure, list, and decision feedback schemes}, 
  year={1968},
  volume={14},
  number={2},
  pages={206-220},
  doi={10.1109/TIT.1968.1054129}}

@ARTICLE{Gallager62,
  author={Gallager, R.},
  journal={IRE Trans. Inf. Theory}, 
  title={Low-density parity-check codes}, 
  year={1962},
  volume={8},
  number={1},
  pages={21-28},
  doi={10.1109/TIT.1962.1057683}}

@INPROCEEDINGS{galligan2023upgrade,
  author={Galligan, Kevin and Yuan, Peihong and Médard, Muriel and Duffy, Ken R.},
  booktitle={IEEE GLOBECOM}, 
  title={Upgrade error detection to prediction with {GRAND}}, 
  year={2023},
  volume={},
  number={},
  pages={1818-1823},
  doi={10.1109/GLOBECOM54140.2023.10437145}}

@inproceedings{galligan2023block,
  author={Galligan, Kevin and Médard, Muriel and Duffy, Ken R.},
  booktitle={CISS}, 
  title={{Block turbo decoding with ORBGRAND}}, 
  year={2023},
}

@INPROCEEDINGS{condo2022_iterative,
    author = {Condo, Carlo},
    title = {Iterative Soft-Input Soft-Output Decoding with Ordered Reliability Bits {GRAND}},
    year = {2022},
    booktitle={IEEE GLOBECOM Workshops}, 
    pages={510-515},
}

@inproceedings{Duffy23ORBGRANDAI,
  author = {Duffy, Ken R. and Grundei, Moritz and M\'edard, Muriel},
  title = {Using channel correlation to improve decoding -- {ORBGRAND-AI}},
  year = {2023},
  booktitle = {IEEE GLOBECOM}
}

@inproceedings{an2023soft,
  title={Soft decoding without soft demapping with {ORBGRAND}},
  author={An, Wei and M{\'e}dard, Muriel and Duffy, Ken R},
  booktitle={IEEE ISIT},
  pages={1080--1084},
  year={2023},
}

@article{ma2024guessing,
  title={Guessing what, noise or codeword?},
  author={Ma, Xiao},
  journal={IEEE ITW},
  year={2024}
}

@article{zheng2024universal,
      title={A Universal List Decoding Algorithm with Application to Decoding of Polar Codes}, 
      author={Xiangping Zheng and Xiao Ma},
      
      journal = {IEEE Trans. Inf. Theory},
    year={2025},
  volume={71},
  number={2},
  pages={975-995},
}

@article{Mur73,
      author = "Allan H.  Murphy",
      title = "Hedging and Skill Scores for Probability Forecasts",
      journal = "J. Appl. Meteorol. and Climatol.",
      year = "1973",
      volume = "12",
      number = "1",
      pages=   "215 - 223",
}

@article{gupta1970partitions,
  title={Partitions--a survey},
  author={Gupta, Hansraj},
  journal={J. of Res. Nat. Bur. Standards-B Math. Sci. B},
  volume={74},
  pages={1--29},
  year={1970}
}

@ARTICLE{yuan2025SOSCL,
  author={Yuan, Peihong and Duffy, Ken R. and Médard, Muriel},
  journal={IEEE Trans. Inf. Theory}, 
  title={Soft-Output Successive Cancellation List Decoding}, 
  year={2025},
  volume={71},
 number = {2},
  pages = {1007--1017}
}

@article{tal2015list,
  title={List decoding of {P}olar codes},
  author={Tal, I. and Vardy, A.},
  journal={IEEE Trans. Inf. Theory},
  volume={61},
  number={5},
  pages={2213--2226},
  year={2015},
  publisher={IEEE}
}

@INPROCEEDINGS{2024ZhengLCGCD,
  author={Zheng, Xiangping and Ma, Xiao},
  booktitle={IEEE ITW}, 
  title={Locally Constrained Guessing Codeword Decoding of Short Block Codes}, 
  year={2024},
  volume={},
  number={},
  pages={442-447}}

@article{pyndiah_1998,
	title = {Near-optimum decoding of product codes: block turbo codes},
	volume = {46},
	issn = {00906778},
	doi = {10.1109/26.705396},
	number = {8},
	fjournal = {IEEE Transactions on Communications},
	journal = {IEEE Trans. Commun.},
	author = {Pyndiah, R.M.},
	year = {1998},
	pages={1003--1010},
}

@article{duffy2025SOGCD,
  title={Soft-output Guessing Codeword Decoding},
  author={Duffy, Ken R and Yuan, Peihong and Griffin, Joseph and M{\'e}dard, Muriel},
  journal= {IEEE Commun. Lett.},
   year={2025},
  volume={29.2},
  number={},
  pages={328-332},
  doi={10.1109/LCOMM.2024.3516486}
}

@article{cocskun2024precoded,
  title={Precoded Polar Product Codes},
  author={Co{\c{s}}kun, Mustafa Cemil},
  journal={IEEE ISIT},
  year={2024}
}

@inproceedings{koopman2004cyclic,
  author={P. {Koopman} and T. {Chakravarty}},
  booktitle={Int. Conf. Dep. Sys. Net.}, 
  title={Cyclic redundancy code ({CRC}) polynomial selection for embedded networks}, 
  year={2004},
  volume={},
  number={}
}

@ARTICLE{OSD95,
  author={Fossorier, M.P.C. and Shu Lin},
  journal={IEEE Trans. Inf. Theory}, 
  title={Soft-decision decoding of linear block codes based on ordered statistics}, 
  year={1995},
  volume={41},
  number={5},
  pages={1379-1396},
}

@INPROCEEDINGS{Rowshan22,
  author={Rowshan, Mohammad and Yuan, Jinhong},
  booktitle={IEEE ISIT}, 
  title={Constrained Error Pattern Generation for {GRAND}}, 
  year={2022},
  volume={},
  number={},
  pages={1767-1772},
}

@INPROCEEDINGS{Rowshan23,
  author={Rowshan, Mohammad and Yuan, Jinhong},
  booktitle={IEEE GLOBECOM}, 
  title={Low-Complexity {GRAND} by Segmentation}, 
  year={2023},
  volume={},
  number={},
  pages={6145-6151},
  }

@INPROCEEDINGS{griffin2024using,
  title={Using a Single-Parity-Check to reduce the Guesswork of Guessing Codeword Decoding},
  author={Griffin, Joseph and Yuan, Peihong and Duffy, Ken R and Medard, Muriel},
  booktitle={CISS},
  year={2025}
}

@ARTICLE{Liuetal23,
  author={Liu, Mengxiao and Wei, Yuejun and Chen, Zhenyuan and Zhang, Wenyi},
  fjournal={IEEE Transactions on Information Theory}, 
  journal={IEEE Trans. Inf. Theory}, 
  title={{ORBGRAND} Is Almost Capacity-Achieving}, 
  year={2023},
  volume={69},
  number={5},
  pages={2830-2840}
}

@inproceedings{Duffy21_ordered,
  title={Ordered reliability bits guessing random additive noise decoding},
  author={Duffy, K. R.},
  booktitle={IEEE ICASSP},
  year={2021},
}

@ARTICLE{Riaz24,
  author={Riaz, Arslan and Yasar, Alperen and Ercan, Furkan and An, Wei and Ngo, Jonathan and Galligan, Kevin and Médard, Muriel and Duffy, Ken R. and Yazicigil, Rabia Tugce},
  fjournal={IEEE Journal of Solid-State Circuits}, 
  journal = {IEEE J. Solid-State Circuits},
  title={A Sub-0.8-{pJ}/bit universal Soft-Detection Decoder Using {ORBGRAND}}, 
  year={2025},
  }

@INPROCEEDINGS{janz2025soft,
  author={Janz, Tim and Obermüller, Simon and Zunker, Andreas and Ten Brink, Stephan},
  booktitle={Int. Symp. Topics Coding}, 
  title={Soft-Output from Covered Space Decoding of Product Codes}, 
  year={2025},
  volume={},
  number={},
  pages={1-5},
  doi={10.1109/ISTC65386.2025.11154647}}

@article{Walter2006StrictProperRules,
author = {Krämer, Walter},
title = {Evaluating probability forecasts in terms of refinement and strictly proper scoring rules},
journal = {J. Forecasting},
volume = {25},
number = {3},
pages = {223-226},
year = {2006}
}

@INPROCEEDINGS{2024WangPCGRAND,
  author={Wang, Yixin and Liang, Jifan and Ma, Xiao},
  booktitle={Int. Conf. Comput. Commun. Sys.}, 
  title={Partially Constrained {GRAND} of Linear Block Codes}, 
  year={2024},
  volume={},
  number={},
  pages={617-622}
}

@article{Lukas2025SYGRAND,
  title={{SOGRAND} Assisted Guesswork Reduction},
  author={Rapp, Lukas and Médard, Muriel and Duffy, Ken R.},
  journal={IEEE Commun. Lett.},
  year={2025}
}

@article{Lukas2025Tree,
  title={A Balanced Tree Transformation to Reduce {GRAND} Queries},
  author={Rapp, Lukas and Feng, Jiewei and Médard, Muriel and Duffy, Ken R.},
  journal={IEEE ISIT},
  year={2025},
  pages={1-6}
}

@ARTICLE{Alex2026LRC,
  author={Mariona, Alexander and Duffy, Ken R. and Médard, Muriel},
  journal={IEEE Transactions on Information Theory}, 
  title={The Linear Reliability Channel}, 
  year={2026},
  volume={72},
  number={8},
  pages={5600-5622},
  doi={10.1109/TIT.2026.3700034}}

\section*{Acknowledgment}

This work was supported by the Defense Advanced Research Projects Agency (DARPA) under Grant HR00112120008, and by the National Science Foundation under Grant ECCS-2433994 and ECCS-2433996

\end{document}